\documentclass[notitlepage, onecolumn]{article}

\usepackage{graphicx}
\usepackage{amsmath, amsthm, mathrsfs}
\usepackage{amssymb}
\usepackage[dvips]{epsfig}
\usepackage[noadjust]{cite}
\usepackage{setspace}
\usepackage[hmarginratio=1:1,top=32mm,columnsep=20pt]{geometry}
\usepackage{hyperref}
\usepackage[noadjust]{cite}

\newtheorem{theorem}{Theorem}
\newtheorem{remark}{Remark}
\newtheorem{assumption}{Assumption}

\newtheorem{lemma}{Lemma}

\newcommand{\bs}{\boldsymbol}

\newcommand{\mbf}{\mathbf}
\newcommand{\be}{\begin{equation}}
\newcommand{\ee}{\end{equation}}

\date{}

\title{Cooperative Output Regulation of Linear Multi-agent Systems with Communication Constraints}

\author{Abdelkader Abdessameud and Abdelhamid Tayebi
\thanks{This work was supported by the Natural Sciences and Engineering Research Council of Canada (NSERC).}
\thanks{The authors are with the Department of Electrical and Computer Engineering, University of Western Ontario, London, Ontario, Canada. The third author is also with the Department of Electrical Engineering, Lakehead University, Thunder Bay, Ontario, Canada.
{\tt\small aabdess@uwo.ca, tayebi@ieee.org} }%
}

\begin{document}
\maketitle

\begin{abstract}
In this paper, we consider the cooperative output regulation problem for heterogeneous linear multi-agent systems in the presence of communication constraints. Under standard assumptions on the agents dynamics, we propose a distributed control algorithm relying on intermittent and asynchronous discrete-time information exchange that can be subject to unknown time-varying delays and information losses. We show that cooperative output regulation can be reached for arbitrary
characteristics of the discrete-time communication process and under mild assumptions on the interconnection topology between agents. A numerical example is given to illustrate the effectiveness of our theoretical results.
\end{abstract}

\section{Introduction}
            Distributed coordination in multi-agent systems has recently gained extensive attention due to its potential applications in engineering, biological and social systems \cite{Ren:Cao:book}. The main goal in distributed coordinated control is to realize a group objective using local interaction between agents. From this perspective, various coordinated control algorithms for identical linear multi-agent systems described by single/mulitple integrators, linear oscillators, and more general high-order dynamics have been proposed in the literature under some specific assumptions on the interconnection between agents. Examples of these results can be found in
            \cite{ren2008distributed, bullo2009distributed, su2009synchronization, abdessameud2010consensus, li2010consensus, yu2010distributed,  cao2012distributed, liu2012necessary, abdessameud2013consensus, li2015designing}, where several methods have been proposed to solve different, yet closely-related, coordination problems including consensus, flocking, formation maintenance, cooperative tracking, and synchronization. For heterogeneous multi-agent systems, \cite{wieland2011internal} have shown that distributed algorithms can also be derived using similar methods combined with results form the classical output regulation theory \cite{francis1976internal}. In fact, the cooperative output regulation problem has emerged as an important problem that encapsulates many coordinated control problems of heterogeneous agents (see, for instance, \cite{xiang2009synchronized, su2012cooperative1, su2012cooperative2, liu2015robust, meng2015coordinated}).

            In this paper, we consider the cooperative output regulation problem of heterogeneous linear multi-agent systems
            governed by the general dynamics
            \begin{eqnarray}\label{model_heterogenous}\begin{array}{lcl}
            \dot{x}_i &=& A_i x_i + B_i u_i + E_i \upsilon  \\
            y_i &=& C_i x_i + D_i u_i + F_i \upsilon
             \end{array},&\quad& i\in\mathcal{N},
            \end{eqnarray}
            where $x_i \in\mathbb{R}^{N_{x_i}}$, $u_i \in\mathbb{R}^{N_{u_i}}$, $y_i \in\mathbb{R}^{N_{y_i}}$ are, respectively, the state vector, the control input, and the measured output of the $i$-th agent, $A_i$, $B_i$, $C_i$, $D_i$, $E_i$, $F_i$ are matrices of appropriate dimensions, and $\mathcal{N}:=\{1,\ldots,n\}$ is the index set of all agents. The signal $\upsilon\in\mathbb{R}^{q}$ models both the global reference signal to be tracked and the disturbance to be rejected by each agent and is generated by the following exogenous dynamic system
                \begin{eqnarray}\label{exosystem}
                \dot{\upsilon} &=& S \upsilon,
                \end{eqnarray}
            with some initial states and $S\in\mathbb{R}^{q\times q}$ being a known matrix. The objective consists in designing appropriate inputs $u_i$ such that stabilization of some regulated error signal, to be defined later, is guaranteed.

            Clearly, in the case where all agents can sense/estimate the exogenous signal $\upsilon$, the above described problem reduces to the output regulation problem of a single plant studied by \cite{huang2004nonlinear}. In this work, we are interested in the case where the exogenous signal $\upsilon$ can be captured only by a group of informed agents whereas the other agents attempt to achieve the control objective by coordinating with other team members. Basically, all agents need to reach an agreement on the estimate of the time-varying exogenous signal using local information exchange performed according to some interconnection topology. To that end, all agents are interconnected in the sense that some information can be transmitted between agents according to some graph topology.

            The above problem, with its variants regarding the system model, has been addressed in \cite{xiang2009synchronized, su2012cooperative1, su2012cooperative2, liu2015robust, meng2015coordinated}, to cite a few, under different assumptions on the interconnection graph, however, under similar idealized assumptions on the interaction between agents which is generally performed using communication networks. In fact, in all the above mentioned papers, the information exchange is assumed ideal in the sense that the information is continuously transmitted between agents and received in real time. In practice, however, communication over networks is subject to time-varying delays, packets dropouts, and can be discrete-time and intermittent due to various environmental and/or technological factors. Motivated by this, our main interest in this paper is to solve the cooperative output regulation problem for system \eqref{model_heterogenous} assuming constrained discrete-time communication between agents.

            The second-order consensus problem for double integrators has been studied, for instance, in \cite{yu2010some, sun2009consensus2, cepeda2011exact, liu2014containment}, in the presence of uniform constant communication delays.
            For identical higher-order linear multi-agents,  \cite{zhou2014consensus} presented a consensus algorithm in the presence of arbitrary large constant communication delays. A similar result was also obtained in \cite{zhou2014consensus} in the case of uniform time-varying delays, however, under some conditions on the delays upper bounds and some restrictions on the dynamics of the agents. A common assumption in the above mentioned delay-robust algorithms is that the information exchange is assumed to be continuous in time and the communication delays are perfectly known.
            In \cite{xiang2014synchronised}, a consensus algorithm for high-order heterogeneous agents has been proposed
             assuming sampled-data information exchange subject to known constant communication delays.

             On the other hand, the authors in \cite{gao2010consensus, Wen:Duan:2012, zhou2012synchronization, wen:ren:2013, huang2014intermittentdelayed} presented consensus algorithms, for linear homogeneous multi-agent systems, in the case of intermittent information exchange between agents. However, communication delays have been considered only in \cite{huang2014intermittentdelayed} dealing with the second-order consensus problem under the assumptions of strong connectivity, periodic intermittent communication, and perfectly known time-varying communication delays. More recently, a small-gain framework has been adopted in \cite{Abdessameud:Polushin:Tayebi:2013:ieeetac} to design distributed algorithms for nonlinear second-order systems in the presence of irregular communication delays. The latter approach has been further developed in \cite{abdessameud2015synchronization, abdessameud2015leader, abdessameud2016distributed} to solve similar problems for second-order systems assuming delayed and (not necessarily periodic) intermittent discrete-time information exchange.

            The main contribution of this paper is a solution to the cooperative output regulation problem for heterogeneous linear multi-agent systems assuming discrete-time, intermittent and asynchronous information exchange, subject to non-uniform and unknown irregular communication delays that can be unbounded. As compared to the relevant literature mentioned above, the present work considers the coordinated control problem of high-order heterogeneous multi-agent systems by taking into account all the above communication constraints simultaneously. Our control objective can be
            reached in the case of a general directed graph, describing the interconnection between agents, for
            arbitrary properties of the communication process that can induce large communication blackouts.
            To the best of our knowledge, there is no coordinated control algorithm in the literature that takes into account (simultaneously) all the above mentioned communication constraints for high-order linear heterogeneous (or identical) multi-agent systems.

\section{Problem Formulation}

\subsection{Model Description}\label{sec:model}

            Consider the $n$ heterogenous agents in \eqref{model_heterogenous} and suppose that the external signal $\upsilon$, generated by \eqref{exosystem}, can be estimated only by some informed agents referred to as leaders. The uninformed agents, acting as followers, do not have access to the exogenous signal. Without loss of generality, let $\mathcal{F}:= \{ 1, \ldots, m\}\subset \mathcal{N}$ and $\mathcal{L}:= \mathcal{N}\setminus \mathcal{F}$, with $0< m < n$, denote the sets of indices corresponding the followers and leaders, respectively.

            Let $\mathcal{G}=(\mathcal{N},\mathcal{E},\mathcal{A})$ denote the directed graph that models the interconnection between agents, where $\mathcal{N}$ is the set of nodes representing the agents, $\mathcal{E}\in\mathcal{N}\times\mathcal{N}$ is the set of pairs of edges, and $\mathcal{A}=[a_{ij}]\in\mathbb{R}^{n\times n}$ is the weighted adjacency matrix. An edge $(j, i)\in\mathcal{E}$, represented by a directed link from node $j$ to node $i$, indicates that agent $i$ can receive information from agent $j$ but not {\it vice versa}. A finite sequence of distinct edges of $\mathcal{G}$ in the form $(j,l_1),(l_1,l_2), \ldots,  (l_p, i)$ is called a directed path from $j$ to $i$. The elements of $\mathcal{A}$ are defined such that $a_{ii}:= 0$, $a_{ij}>0$ if $(j,i)\in\mathcal{E}$, and $a_{ij}=0$ if $(j,i)\notin \mathcal{E}$. The Laplacian matrix $\mathbf{L}:=[l_{ij}]\in\mathbb{R}^{n\times n}$ associated to $\mathcal{G}$ is defined such that: $l_{ii}=\sum_{j=1}^n a_{ij}$, and $l_{ij}=-a_{ij}$ for $i\neq j$.

            Also, the information exchange is discrete in time and subject to irregular communication delays. More precisely, for each $(j,i)\in\mathcal{E}$, agent $j$ can send data to agent $i$ only at instants $t_{k_{ij}} = k_{ij} t_s$, with $k_{ij} \in \mathcal{S}_{ij}\subseteq \{0,1, 2, \ldots \}$ and  $t_s$ being a common sampling period. This information can be received by agent $i$ at instant $t_{k_{ij}}+\tau_{k_{ij}}$, where $\left(\tau_{k_{ij}}\right)_{k_{ij} \in \mathcal{S}_{ij}}$ is a sequence of communication delays that take values in ${\mathbb R}_+\cup\{ +\infty\}$, where  $\tau_{k_{ij}}=+\infty$ means that the corresponding data has been lost during transmission.

\subsection{Problem Statement and Assumptions}\label{sec:pb:statement}

            Our objective consists in designing a control algorithm for \eqref{model_heterogenous}-\eqref{exosystem} such that the regulated error signal $e_i\in\mathbb{R}^{p_{e_i}}$ written in the general form \vspace{-0.05 in}
            \begin{eqnarray}\label{regulated error}
                e_i  &=& C_{e_i} x_i + D_{e_i} u_i + F_{e_i} \upsilon,\quad i\in\mathcal{N},\vspace{-0.05 in}
            \end{eqnarray}
            satisfies $e_{i}(t) \to 0$, for $i\in\mathcal{N}$, for arbitrary initial conditions. Note that the above formulation, with $\mathcal{L}\neq \mathcal{N}$, is general in the sense that it captures many problems relevant to multi-agent systems such as leader-follower consensus/synchronization and cooperative tracking/disturbance rejection.

%
\begin{assumption}\label{assum:control:observer}
            \begin{itemize}
             \item [i)] $(A_i, B_i)$ is stabilizable for $i\in\mathcal{N}$.
             \item [ii)] $(C_i, A_i)$, for $i\in\mathcal{F}$, and $(\bar{C}_i, \bar{A}_i )$, for $i\in\mathcal{L}$, are detectable, with
                 \begin{eqnarray}\label{bar_A}
                 \bar{A}_i:= \left[ \begin{array}{ll} A_i & E_i\\ 0_{q\times N_{x_i}}& S\end{array}\right], \quad \bar{C}_i := \left[C_i \quad F_i\right], \quad i\in\mathcal{L}.\end{eqnarray}
             \item [iii)] For all $i\in\mathcal{N}$, there exist matrices  $\Pi_i,~ \Gamma_i$ such that
                    \begin{eqnarray}\label{regulator:equations}
                    \begin{array}{lcl}\Pi_i S &=& A_i \Pi_i + B_i \Gamma_i + E_i\\
                    0&=& C_{ei} \Pi_i + D_{ei} \Gamma_i + F_{ei}\end{array}.
                    \end{eqnarray}
                    \end{itemize}
            \end{assumption}

            Assumption~\ref{assum:control:observer}, item $iii)$ is standard and necessary for the solvability of the output regulation problem \cite{huang2004nonlinear}. Item $ii)$ in the above assumption distinguishes the two sets $\mathcal{L}$ and $\mathcal{F}$, and implies that the leaders can estimate their states as well as the external signal $\upsilon$ using their measured outputs. This is not the case for the follower agents where only their states are detectable from the measurements.

            \begin{assumption}\label{assum:S}
            The eigenvalues of $S$ in \eqref{exosystem} lie on the imaginary axis of the complex plane.
            \end{assumption}
%
%
            \begin{assumption}\label{AssumptionCommunicationBlackouts}
                    For each pair $(j,i)\in\mathcal{E}$, there exist a strictly increasing infinite subsequence $\bar{\mathcal{S}}_{ij} = \{k_{ij}^{(1)}, k_{ij}^{(2)}, \ldots \}\subseteq \mathcal{S}_{ij}$  and $h^*>0$  such that: $t_{k_{ij}^{(l+1)}} +\tau_{k_{ij}^{(l+1)}} - t_{k_{ij}^{(l)}}\le h^*$, for $l=1, 2,\ldots .$
                                \end{assumption}
%

            Assumption~\ref{AssumptionCommunicationBlackouts} states that, for each pair $(j,i)\in\mathcal{E}$, there exists a subsequence of transmission time instants $\bar{\mathcal{S}}_{ij}$ and the corresponding communication delays such that the information sent by agent $j$ at instants $t_{k_{ij}}$ for  $k_{ij} \in \bar{\mathcal{S}}_{ij}$ are successfully received by agent $i$. Further, for each pair $(j,i)\in\mathcal{E}$, the maximum length of communication blackout intervals between agents $j$ and $i$ does not exceed an arbitrary (not necessarily known) bound $h^*$. Obviously, an infinite $h^*$ implies that communication between agents is completely lost and the cooperative output regulation problem described above may not be solved, in general. Also, note that the subsequences $\mathcal{S}_{ij}$ and $\bar{\mathcal{S}}_{ij}$ are defined for each edge in $\mathcal{E}$, which shows that the intermittent and delayed discrete-time information exchange described above is also asynchronous.
            \begin{assumption}\label{assumption_graph}
                For each node $i\in\mathcal{L}$, the edge $(j,i)\notin \mathcal{E}$ for all $j\in\mathcal{N}$. Also, for each node $i\in\mathcal{F}$, there exists at least one node $j\in\mathcal{L}$ such that a directed path from $j$ to $i$ exists in $\mathcal{G}$.  \vspace{-0.1 in}
            \end{assumption}
            Assumption~\ref{assumption_graph} implies that a leader node does not receive information from any other node in $\mathcal{G}$, which is reasonable since each leader can estimate the external signals using only its measurements. Also, for each follower agent, there exists at least one leader having a directed path to that follower. Consequently, the Laplacian matrix  associated to $\mathcal{G}$ takes the form
         \begin{eqnarray}
         \label{laplacian}
          \mathbf{L}&=&\left[\begin{array}{ll} \mathbf{L}_1&\mathbf{L}_2\\ 0_{(n-m)\times m}&0_{(n-m)\times (n-m)}\end{array}\right].
         \end{eqnarray}
        and satisfies the properties in the following lemma.
        \begin{lemma}\label{lemma_Mei}\cite[Lemma 2.3]{Mei:Ren:2012}
         Consider $\mathbf{L}$ defined in \eqref{laplacian}. Under Assumption~\ref{assumption_graph}, the matrix $\mathbf{L}_1$  is a nonsingular M-matrix\footnote{A matrix $A \in\mathbb{R}^{n\times n}$ is said to be a nonsingular M-matrix if $A \in \mathcal{Z}_n$ and all eigenvalues of $A$ have positive real parts, where $\mathcal{Z}_n \subset \mathbb{R}^{n\times n}$ denotes the set of square matrices with non-positive off-diagonal entries \cite{qu2009cooperative}.}, each entry of $ -\mathbf{L}^{-1}_1 \mathbf{L}_2$ is nonnegative, and all row sums of $-\mathbf{L}_1^{-1} \mathbf{L}_2$ are equal to one.
        \end{lemma}

\section{Distributed Output Regulation Algorithm}

           Consider the following control law for each agent
           \begin{eqnarray}
            \label{input:control} u_i &=& K_i (\hat{x}_i - \Pi_i \hat{\upsilon}_i) + \Gamma_i \hat{\upsilon}_i, \\
            \label{observer}  \dot{\hat{x}}_i &=& A_i \hat{x}_i + E_i \hat{\upsilon}_i + B_i u_i + L_{1i}(\hat{y}_i - y_i),\\
            \label{obeserver:output}\hat{y}_i &=& C_i \hat{x}_i + F_i \hat{\upsilon}_i + D_i u_i,
            \end{eqnarray}
            where $\hat{x}_i$ is the observed state, $\hat{y}_i$ is the observer output, $K_i$, $L_{1i}$ are gain matrices with appropriate dimensions, and the signal $\hat{\upsilon}_i$ is an estimate of the exogenous signal obtained by each agent according to the following algorithm
           \begin{eqnarray}\label{observer:reference}
            \dot{\hat{\upsilon}}_i &=& S \hat{\upsilon}_i +  \eta_i, \quad \mbox{for}~i\in\mathcal{N},
            \end{eqnarray}
            where $\eta_i\in\mathbb{R}^q$, for $i\in\mathcal{F}$, is an input to be designed, and
             \begin{eqnarray}\label{eta_leaders}
            \eta_i &=&  L_{2i}(\hat{y}_i - y_i), \quad \mbox{for}~i\in\mathcal{L},
            \end{eqnarray}
            with $L_{2i}$, $i\in\mathcal{L}$, being a gain matrix of appropriate dimension. Let $\bar{L}_i = \begin{bmatrix} L_{1i}^\top & L_{2i}^\top \end{bmatrix}^\top$, for $ i\in\mathcal{L}$.

             Note that \eqref{input:control}-\eqref{eta_leaders}, for $i\in\mathcal{L}$, is a classical observer-based control algorithm that guarantees, along with Assumption~\ref{assum:control:observer}, the exponential convergence to zero of the regulated error for $i\in\mathcal{L}$ \cite{huang2004nonlinear}. Also, as it will become clear later, the control scheme \eqref{input:control}-\eqref{obeserver:output}, for $i\in\mathcal{F}$, ensures the solvability of the output regulation problem for all follower agents provided that each follower can estimate the state of the exogenous system with an appropriate design of $\eta_i$ in \eqref{observer:reference}, $i\in\mathcal{F}$,  using intermittent, delayed and discrete-time communication.

             For this purpose, we suppose that the data that can be transmitted from agent $j$ to agent $i$ at instant $t_{k_{ij}}=k_{ij}t_s$, for each $(j,i)\in\mathcal{E}$ and each $k_{ij}\in\mathcal{S}_{ij}$, consists of the sequence number $k_{ij}$ of a transmission instant $t_{k_{ij}}$, and the vector $\hat{\upsilon}_j(k_{ij} t_s)$ obtained from \eqref{observer:reference} for $j\in\mathcal{N}$. Also, for each pair $(j,i)\in\mathcal{E}$ and each time instant $t\geq 0$, let $k_{ij}^\mathrm{x}(t)$ denote the largest integer number such that $\hat{\upsilon}_j(t_{k_{ij}^\mathrm{x}(t)})$, with $t_{k_{ij}^\mathrm{x}(t)} = k_{ij}^\mathrm{x}(t) t_s$, is the most recent information of agent $j$ that is already delivered to agent $i$ at $t$. This integer $k_{ij}^\mathrm{x}(t)$ can be determined by a simple comparison of the received sequence numbers.

            Consider the following input $\eta_i$ in \eqref{observer:reference} for $i\in\mathcal{F}$
            \begin{equation}\label{eta_followers}
            \eta_i = - \sum_{j=1}^n a_{ij} \big(\hat{\upsilon}_i - \underbrace{e^{S(t-t_{k_{ij}^\mathrm{x}(t)})} \hat{\upsilon}_j(t_{k_{ij}^\mathrm{x}(t)})}_{\hat{\upsilon}_{ij}^*}\big),
            \end{equation}
                where $\hat{\upsilon}_{ij}^*$ can be regarded as a prediction of the signal $\hat{\upsilon}_{j}(t)$ obtained using the most recent information received from agent $j$.

              \begin{theorem}\label{theorem1}
                Consider the multi-agent system \eqref{model_heterogenous} with \eqref{exosystem} and suppose that Assumptions \ref{assum:control:observer}-\ref{assumption_graph} hold. For each agent, consider the control algorithm \eqref{input:control}-\eqref{observer:reference} where the input $\eta_i$ is given in \eqref{eta_leaders} for $i\in\mathcal{L}$ and in \eqref{eta_followers} for $i\in\mathcal{F}$. Pick the gains $K_i$, $L_{1i}$ and $L_{2i}$ such that $(A_i + B_i K_i)$, for $i\in\mathcal{N}$, $(A_i + L_{1i}C_i)$, for $i\in\mathcal{F}$, and
                $\left(\bar{A}_i + \bar{L}_i\bar{C}_i\right)$, for $i\in\mathcal{L}$, are stable matrices. Then, the cooperative output regulation problem is solved for arbitrary initial conditions.
                          \end{theorem}
        \begin{proof}
            Define the following error signals: \begin{equation*}\varepsilon_i := x_i - \Pi_i \hat{\upsilon}_i, \quad \tilde{x}_i := \hat{x}_i - x_i, \quad \tilde{\upsilon}_i := \hat{\upsilon}_i - \upsilon,\end{equation*} for $i\in\mathcal{N}$. Using \eqref{model_heterogenous}-\eqref{exosystem} and \eqref{input:control}-\eqref{eta_leaders}, and taking into account  point $iii)$ in Assumption~\ref{assum:control:observer}, the regulated error signal, in \eqref{regulated error}, can be shown to satisfy
             \begin{eqnarray}
             \label{error:all}
                e_i &=& (C_{e_i} + D_{e_i} K_i) \varepsilon_i + D_{e_i} K_{i} \tilde{x}_i - F_{e_i} \tilde{\upsilon}_i,\\
            \label{dot:epsilon}
            \dot{\varepsilon}_i &=& (A_i + B_i K_i)\varepsilon_i + B_i K_i \tilde{x}_i - E_i \tilde{\upsilon}_i - \Pi_i \eta_i,
            \end{eqnarray}
            for $i\in\mathcal{N}$, with
                \begin{eqnarray}\label{closed:observer:leaders}
              \left[ \begin{array}{c} \dot{\tilde{x}}_i \\ \dot{\tilde{\upsilon}}_i
              \end{array}\right] = \left(\bar{A}_i + \bar{L}_i\bar{C}_i \right)\left[ \begin{array}{c} \tilde{x}_i \\ \tilde{\upsilon}_i
              \end{array}\right], \quad i\in\mathcal{L},
            \end{eqnarray}
            \begin{eqnarray} \label{closed:observer:followers}
            \dot{\tilde{x}}_i &=& (A_i + L_{1i}C_i)\tilde{x}_i + (E_i + L_{1i} F_i) \tilde{\upsilon}_i,\\
            \label{dot_tilde_upsilon_followers}\dot{\tilde{\upsilon}}_i &=& S \tilde{\upsilon}_i + \eta_i,
            \end{eqnarray}
            where the two last relations hold for $i\in\mathcal{F}$, matrices
            $\bar{A}_i$, $\bar{C}_i$ are given in \eqref{bar_A}, and we used the relation $\eta_i = L_{2i}(\hat{y}_i - y_i) = L_{2i}(C_i \tilde{x}_i + F_i \tilde{\upsilon}_i)$ for $i\in\mathcal{L}$. \vspace{0.05 in}

            It is straightforward to verify that system \eqref{closed:observer:leaders} is exponentially stable with the above described choice of the gain matrices $K_i$, $L_{1i}$ and $L_{2i}$, for $i\in\mathcal{L}$. It is also easy to verify that each system  \eqref{dot:epsilon}, $i\in\mathcal{N}$, is input-to-state stable (ISS) with respect to the inputs $\eta_i$, $\tilde{\upsilon}_i$ and $\tilde{x}_i$, $i\in\mathcal{N}$. This, with \eqref{error:all} and the fact that $\tilde{\upsilon}_i(t)\to 0$ and $\tilde{x}_i(t)\to 0$ exponentially, lead to the conclusion that $e_i(t)\to 0$ for each $i\in\mathcal{L}$.

            Now, consider system \eqref{observer:reference} with \eqref{eta_followers}, for $i\in\mathcal{F}$, which, using \eqref{exosystem}, can be shown to satisfy
            \begin{equation}\label{tilde_upsilon_proof}
            \dot{\tilde{\upsilon}}_i =
            S \tilde{\upsilon}_i - \sum_{j=1}^n a_{ij} \big(\tilde{\upsilon}_i - e^{S(t-t_{k_{ij}^\mathrm{x}(t)})} \tilde{\upsilon}_j(t_{k_{ij}^\mathrm{x}(t)})\big), \quad i\in\mathcal{F}.
            \end{equation}

            Consider also the the change of coordinates $\bar{\upsilon}_i = V \tilde{\upsilon}_i $, for all $i\in\mathcal{N}$, where $V\in\mathbb{R}^{q\times q}$ is a real orthogonal matrix such that $V S V^{\top} = T$, with $T$ being the real Schur form of $S$. Note that such a canonical form exists for any real square matrix, and $T\in\mathbb{R}^{q\times q}$ is a block upper triangular matrix of the form
               \begin{equation}
               T = \begin{bmatrix}\label{T} T_{11}& T_{12}& \ldots & T_{1p} \\  & T_{22} & \ldots& T_{2p} \\ &&\ddots&\vdots\\ &&&T_{pp}\end{bmatrix},
               \end{equation}
            where $T_{\ell \ell}\in\mathbb{R}^{q_\ell \times q_{\ell}}$, $\ell = 1,\ldots p$, with $q_\ell$ can be equal to either $1$ or $2$ and $\sum_{\ell = 1}^p q_{\ell} = q$, all the elements below $T_{\ell\ell}$ are zeros, and $T_{\ell \hbar}$, for $\ell = 1,\ldots, p-1$ and $\hbar=\ell+1,\ldots, p$, are of appropriate dimensions. Accordingly, $T_{\ell \ell}$, $\ell = 1,\ldots p$, can be either a real number equal to a real eigenvalue of $S$, or a real 2-by-2 matrix having a pair of complex conjugate eigenvalues of $S$. Therefore, in view of Assumption~\ref{assum:S}, we have $T_{\ell \ell} = 0$ if $q_\ell = 1$, and the two eigenvalues of $T_{\ell\ell}$ are complex with zero real parts if $q_\ell = 2$.

            Then, using \eqref{tilde_upsilon_proof}, one can show that
            \begin{equation}\label{bar_upsilon_follower}
            \dot{\bar{\upsilon}}_i = T \bar{\upsilon}_i -\sum_{j=1}^{n} a_{ij}(\bar{\upsilon}_i - e^{T(t-t_{k_{ij}^\mathrm{x}(t)})} \bar{\upsilon}_j(t_{k_{ij}^\mathrm{x}(t)})),
            \end{equation}
            for $i\in\mathcal{F}$, where
             \begin{equation}\label{exp_T}
            e^{T\varsigma} = \begin{bmatrix}
            e^{T_{11}\varsigma}& F_{12}(\varsigma)& \ldots & F_{1p}(\varsigma) \\  & e^{T_{22}\varsigma} & \ldots& F_{2p}(\varsigma) \\ &&\ddots&\vdots\\ &&&e^{T_{pp}\varsigma}
            \end{bmatrix}
            \end{equation}
            and the functions $F_{\ell \hbar}(\varsigma)$, for $\ell=1,\ldots, p-1$ and $\hbar = \ell+1, \ldots, p$, are continuous functions that can be determined, however, their explicit expressions are not needed in the subsequent analysis.

            In view of the upper triangular form of systems \eqref{bar_upsilon_follower}, we let $\bar{\upsilon}_i^{(\ell)}\in\mathbb{R}^{q_\ell}$, $\ell = 1,\ldots, p$, denote the $\ell$-th component of $\bar{{x}}_i$ corresponding to $T_{\ell\ell}$, for $i\in\mathcal{N}$.
            Therefore, one can show from \eqref{T}-\eqref{exp_T} that
            \begin{align} \label{bar_upsilon_follower2}
             \dot{\bar{\upsilon}}_i^{(\ell)} =&~ T_{\ell\ell} \bar{\upsilon}_i^{(\ell)} - \kappa_i\bar{\upsilon}_i^{(\ell)} + \sum_{j=1}^{m} a_{ij}e^{T_{\ell\ell}(t-t_{k_{ij}^\mathrm{x}(t)})} \bar{\upsilon}_j^{(\ell)}(t_{k_{ij}^\mathrm{x}(t)})+ \phi_{i,\ell},
             \end{align}
             with
            \begin{align}\label{zeta_ell}
             \phi_{i,\ell}&=~\sum_{j=m+1}^{n} a_{ij}e^{T_{\ell\ell}(t-t_{k_{ij}^\mathrm{x}(t)})} \bar{\upsilon}_j^{(\ell)}(t_{k_{ij}^\mathrm{x}(t)})+ \sum_{\substack{\hbar = \ell + 1\\ \ell <p}}^{p} \Big(T_{\ell \hbar}\bar{\upsilon}_{i}^{(\hbar)} + \sum_{j=1}^n a_{ij} F_{\ell\hbar}(t- t_{k_{ij}^{\mathrm{x}}(t)}) \bar{\upsilon}_j^{(\hbar)}(t_{k_{ij}^{\mathrm{x}}(t)})\Big)
            \end{align}
            for $\ell = 1, \ldots, p$ and $i\in\mathcal{F}$, where $\kappa_i:=\sum_{j=1}^{n} a_{ij}$, $i\in\mathcal{F}$. Note that $\kappa_i>0$ for $i\in\mathcal{F}$ by Assumption~\ref{assumption_graph}.

            For $\ell\in\{1,\ldots, p\}$, let $\prod_\ell$ denote the system that consists of all $m$ interconnected systems in \eqref{bar_upsilon_follower2}, $i\in\mathcal{F}$, with the vector $\phi_{i,\ell}$ being considered as a perturbation term for each system. The properties of the states of each system $\prod_\ell$ are characterized in the following lemma proved in the Appendix. \vspace{0.05 in}

            \begin{lemma}\label{lemma:proof}
            Consider the above defined system $\prod_\ell$, for some $\ell\in\{1,\ldots, p\}$ and for $i\in\mathcal{F}$. Suppose that the assumptions of Theorem~\ref{theorem1} hold. Then, $\bar{\upsilon}_i^{(\ell)}$ is uniformly bounded and $\bar{\upsilon}_i^{(\ell)}(t)\to 0$, $i\in\mathcal{F}$, provided that $\phi_{i,\ell}$, $i\in\mathcal{F}$, is uniformly bounded and converges asymptotically to zero.
            \end{lemma}
                       Now, we  apply Proposition~\ref{lemma:proof} to each system $\prod_\ell$, for $\ell = p, \ldots, 1$.
            Consider system $\prod_p$ and notice from \eqref{zeta_ell} that
            $$\phi_{i,p}(t) = \sum_{j=m+1}^{n} a_{ij}e^{T_{pp}(t-t_{k_{ij}^\mathrm{x}(t)})} \bar{\upsilon}_j^{(p)}(t_{k_{ij}^\mathrm{x}(t)}).$$
            Since we have already shown that $\tilde{\upsilon}_i(t)\to 0$ exponentially for all $i\in\mathcal{L}$, we know that $\bar{\upsilon}_i(t)\to 0$ for all $i\in\mathcal{L}$. This, with the fact that $(t - t_{k_{ij}^{\mathrm{x}}(t)})\leq h^*$ by Assumption~\ref{AssumptionCommunicationBlackouts}, leads to the conclusion that $\phi_{i,p}(t)\to 0$ for $i\in\mathcal{F}$. Then, using the result of Proposition~\ref{lemma:proof}, we can show that $\bar{\upsilon}_i^{(p)}$ is uniformly bounded and $\bar{\upsilon}_i^{(p)}(t)\to 0$, $i\in\mathcal{F}$.

            For system $\prod_{(p-1)}$,  one can verify from \eqref{zeta_ell} that
            \begin{align*}
                \phi_{i,(p-1)}&= \sum_{j=m+1}^{n} a_{ij}e^{T_{(p-1)(p-1)}(t-t_{k_{ij}^\mathrm{x}(t)})} \bar{\upsilon}_j^{(p-1)}(t_{k_{ij}^\mathrm{x}(t)}) + T_{(p-1)p}\bar{\upsilon}_{i}^{(p)} + \sum_{j=1}^n a_{ij} F_{(p-1) p}(t- t_{k_{ij}^{\mathrm{x}}(t)}) \bar{\upsilon}_j^{(p)}(t_{k_{ij}^{\mathrm{x}}(t)}),
            \end{align*}
            for $i\in\mathcal{F}$. Since all the functions $F_{\ell \hbar}(\varsigma)$,  for $\ell=1,\ldots, p-1$ and $\hbar = \ell+1, \ldots, p$, are continuous, $t_{k_{ij}^{\mathrm{x}}(t)}\to +\infty$ and $(t - t_{k_{ij}^{\mathrm{x}}(t)})$ is bounded, it can be deduced that $\phi_{i,(p-1)}$ is uniformly bounded and $\phi_{i,(p-1)}(t)\to 0$, $i\in\mathcal{F}$. Then, using Proposition~\ref{lemma:proof}, one can show, following the same arguments as above (the case $\ell = p$), that $\bar{\upsilon}_i^{(p-1)}$ is uniformly bounded and $\bar{\upsilon}_i^{(p-1)}(t)\to 0$, $i\in\mathcal{F}$. Similarly, exploiting the above results and the expression of $\phi_{i,(p-2)}$ in \eqref{zeta_ell}, with $\ell = p-2$, one can show that $\phi_{i,(p-2)}$ is uniformly bounded and $\phi_{i,(p-2)}(t)\to 0$, for $i\in\mathcal{F}$. Repeating these steps for $\ell = p-2, \ldots, 1$, one can show that $\bar{\upsilon}_i$ is uniformly bounded and $\bar{\upsilon}_i(t)\to 0$, for $i\in\mathcal{F}$. Consequently, $\tilde{\upsilon}_i$ is uniformly bounded and $\tilde{\upsilon}_i(t)\to 0$, for $i\in\mathcal{F}$.

            In addition, the vector $\eta_i$ in \eqref{eta_followers} can be rewritten as
            $$\eta_i = -\sum_{j=1}^n a_{ij} \big(\tilde{\upsilon}_i - e^{S(t-t_{k_{ij}^\mathrm{x}(t)})} \tilde{\upsilon}_j(t_{k_{ij}^\mathrm{x}(t)})\big),\quad i\in\mathcal{F}.$$ Using the above results and the fact that $(t - t_{k_{ij}^{\mathrm{x}}(t)})\leq h^*$, one can verify that $\eta_i$ is uniformly bounded and $\eta_i(t)\to 0$, for all $i\in\mathcal{F}$. 

            Finally, it can be verified that system \eqref{closed:observer:followers} is ISS with respect to the input $\tilde{\upsilon}_i$, $i\in\mathcal{F}$. Therefore, $\tilde{x}_i$ is uniformly bounded and $\tilde{x}_i(t)\to 0$ for all $i\in\mathcal{F}$. This, with the ISS property of system \eqref{dot:epsilon}, lead to the conclusion that $\varepsilon_i$, $e_i$ are uniformly bounded and $\varepsilon_i(t)\to 0$, $e_i(t)\to 0$, $i\in\mathcal{F}$. The proof is complete.
\end{proof}

              \begin{remark} Note that the result in Theorem~\ref{theorem1} holds provided that  the instants of time at which two successive information received by agent $i$, for each $(j,i)\in \mathcal{E}$, are not spaced by more than an unknown bound $h^*$ that can take arbitrary large values. Also, the proposed approach in this section can be adapted, with obvious modifications, to solve other problems considered in the available literature of heterogeneous multi-agent systems ({\it e.g.}, \cite{su2012cooperative1, xiang2014synchronised, liu2015robust,  haghshenas2015containment}) under the same communication constraints. For instance, in the special case where each agent can measure its state vector and the leader agents have direct access to the exogenous signal; $y_i = \mbox{col}\{x_i, \upsilon\}$ for $i\in\mathcal{L}$ and $y_i = x_i$ for $i\in\mathcal{F}$, the output regulation problem is equivalent to the one studied in \cite{su2012cooperative1} in the case of delay-free communication between agents. In this case, one can consider the control  $u_i = K_i (x_i - \Pi_i \hat{\upsilon}_i) + \Gamma_i \hat{\upsilon}_i$, with $\hat{\upsilon}_i=\upsilon$ for $i\in\mathcal{L}$ and $\hat{\upsilon}_i$, for $i\in\mathcal{F}$, is given by \eqref{observer:reference} with \eqref{eta_followers}. Following similar steps in the proof of our main results below, it can be shown that the cooperative output regulation problem in this case is solved under similar assumptions in Theorem~\ref{theorem1}.
             \end{remark}

\section{A Numerical Example}
We consider a multi-agent system consisting of six agents modeled as in \eqref{model_heterogenous}-\eqref{exosystem} with
        \begin{eqnarray}
        A_i = \begin{bmatrix} 0&1\\ a_i&a_i\end{bmatrix},&& B_i=\begin{bmatrix} 0\\1\end{bmatrix},C_i=\begin{bmatrix}1&0\end{bmatrix},~D_i = 0, \nonumber\\
        E_i = \begin{bmatrix}0&0\\0&b_i\end{bmatrix}, && F_i = \begin{bmatrix}c_i&0\end{bmatrix},~ S=\begin{bmatrix}0&1\\-1&0\end{bmatrix},\nonumber
        \end{eqnarray}
        where $(a_i,b_i)=(-2,0)$, $i=1,2$, $(a_i,b_i)=(0,1)$, $i=3,4,5,6$, $c_i=-1$, $i=5,6$. Accordingly, we define $\mathcal{F}:=\{1,\ldots,4\}$ and $\mathcal{L}:=\{5,6\}$. The state of each agent is denoted by $x_i := (x_{1_i}, x_{2_i})^\top\in\mathbb{R}^2$, $i\in\mathcal{N}$, where $x_{1_i}$ and $x_{2_i}$ are, respectively, the position and velocity of the agent. The state of the exogenous system \eqref{exosystem} is denoted by $\upsilon = (r, w)^\top$. The regulated error in \eqref{regulated error} is selected as: $C_{e_i}=C_i$, $D_{e_i}=0$, $F_{e_i}=[-1~0]$, \textit{i.e.,} it is required that the position of each agent converges to the reference position defined by $r$; $e_i = x_{1_i}-r$. We can verify that Assumption~\ref{assum:control:observer} is satisfied, in particular, equations \eqref{regulator:equations} admit a solution given by: $\Pi_i = I_2$ and $\Gamma_i = \left[\begin{array}{cc}-(1+a_i)&-(a_i+b_i)\end{array}\right]$, $i\in\mathcal{N}$.

        The Laplacian matrix associated to $\mathcal{G}$ is written as in \eqref{laplacian} with
        \begin{eqnarray}
        \renewcommand*{\arraystretch}{0.75}
        \mbf{L}_1 = \begin{bmatrix} 2& 0&0 &-1\\
        -1& 1& 0& 0\\-1& 0& 2& -1\\0& -1& 0& 2\end{bmatrix},\quad \mbf{L}_2=\begin{bmatrix}-1&0\\0 &0 \\0 &0\\0& -1\end{bmatrix},\nonumber
        \end{eqnarray}
        and the communication process between agents is described in Section~\ref{sec:model}, with $h^* = 1.5~\sec$.

        Now, the control gains are selected as described in Theorem~\ref{theorem1}: $K_i = \Big[-10~~-8\Big]$, for $i\in\mathcal{N}$, $L_{1_i}=\Big[-15~~-25\Big]^\top$, $L_{2_i}=\Big[-10~~-10\Big]^\top$ for $i\in\mathcal{L}$, $L_{1_i}=\Big[-10~~-10\Big]^\top$ for $i\in\mathcal{F}$. The obtained results when applying the control law in Theorem~\ref{theorem1} are shown in Figs.~\ref{fig:error}-\ref{fig:positions}. It can be seen from Fig.~\ref{fig:error} that the regulated error signals of all agents converge to zero. This can also be seen from Fig.~\ref{fig:positions} which shows that the position-like states of all agents converge to the reference signal $r(t)$ despite the presence of large communication blackout intervals between each pair of communicating agents.

\begin{figure}[t]
\centering
\begin{minipage}[h]{0.49\columnwidth}
\centering
\includegraphics[width=1.1\columnwidth]{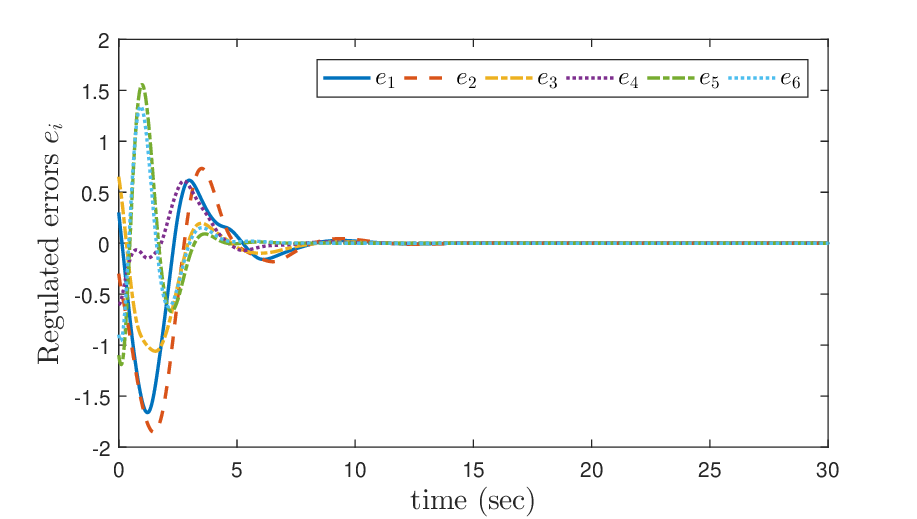}
\caption{\footnotesize The regulated errors of the agents.} \label{fig:error}
\end{minipage}
 \hfill
\begin{minipage}[h]{0.49\columnwidth}
\centering
\includegraphics[width=1.1\columnwidth]{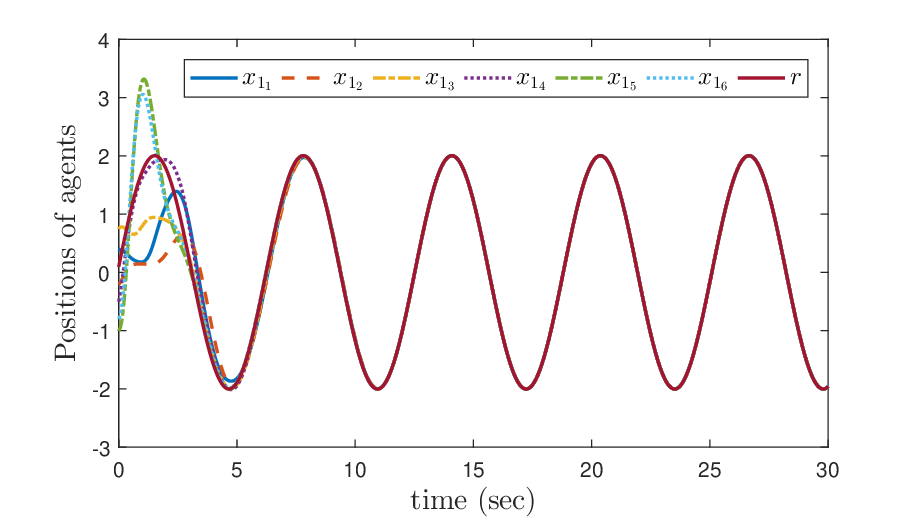}
\caption{\footnotesize Position-like states of all agents.} \label{fig:positions}
\end{minipage}
\end{figure}

          %
%

\section{Conclusion}
        We considered the cooperative output regulation problem of high-order heterogeneous multi-agent systems with constrained discrete-time information exchange. The problem has been solved under mild assumptions on the directed interconnection graph topology, with intermittent and asynchronous discrete-time information exchange in the presence of unknown time-varying communication delays and communication blackouts.

\appendix

\section{Proof of Lemma~\ref{lemma:proof}}\label{appen:proof:lemma}

            Consider system $\prod_\ell$, given in \eqref{bar_upsilon_follower2} for some $\ell\in\{1,\ldots, p\}$ and $i\in\mathcal{F}$;
            \begin{eqnarray*}
             \dot{\bar{\upsilon}}_i^{(\ell)} &=& T_{\ell\ell} \bar{\upsilon}_i^{(\ell)} - \kappa_i\bar{\upsilon}_i^{(\ell)} + \sum_{j=1}^{m} a_{ij}e^{T_{\ell\ell}(t-t_{k_{ij}^\mathrm{x}(t)})} \bar{\upsilon}_j^{(\ell)}(t_{k_{ij}^\mathrm{x}(t)})+ \zeta_{i,\ell}.
             \end{eqnarray*}

              In view of Assumption~\ref{assum:S} and \eqref{T}, it can be deduced that $T_{\ell\ell}=0$ if $q_{\ell}=1$, and the two eigenvalues of $T_{\ell\ell}$ are complex conjugates with zero real parts if $q_{\ell}=2$, for each $\ell= 1,\ldots, p$. Consider the change of variables $z_i^{(\ell)}(t) = e^{-T_{\ell\ell}(t-t_0)}\bar{\upsilon}_i^{(\ell)}(t)$ and $\xi_{i,\ell}(t) = e^{-T_{\ell\ell}(t-t_0)}\zeta_{i,\ell}(t)$, for all $t\geq t_0$ and $i\in\mathcal{F}$. Using the relation $\dot{z}_i^{(\ell)}= e^{-T_{\ell\ell}(t-t_0)}\dot{\bar{\upsilon}}_i^{(\ell)} - T_{\ell\ell}e^{-T_{\ell\ell}(t-t_0)}{\bar{\upsilon}}_i^{(\ell)}$, we can verify that
           \begin{eqnarray}\label{dot_z}
           \dot{z}_i^{(\ell)}
           &=& e^{-T_{\ell\ell}(t-t_0)}\Big( - \kappa_i\bar{\upsilon}_i^{(\ell)} + \sum_{j=1}^{m} a_{ij}e^{T_{\ell\ell}(t-t_{k_{ij}^\mathrm{x}(t)})} \bar{\upsilon}_j^{(\ell)}(t_{k_{ij}^\mathrm{x}(t)})\Big) \nonumber\\
           &&+ e^{-T_{\ell\ell}(t-t_0)} \big(T_{\ell\ell} \bar{\upsilon}_i^{(\ell)} + \zeta_{i,\ell}\big) - T_{\ell\ell}e^{-T_{\ell\ell}(t-t_0)}{\bar{\upsilon}}_i^{(\ell)}   \nonumber\\
           &=& - \kappa_iz_i^{(\ell)} + \sum_{j=1}^{m} a_{ij}z_j^{(\ell)}(t_{k_{ij}^{\mathrm{x}}(t)}) + \xi_{i,\ell},
           \end{eqnarray}
           for $i\in\mathcal{F}$. Therefore, the following estimate
           \begin{equation}\label{estimate_state}
           |z_i^{(\ell)}(t)| \leq e^{-\kappa_i(t-t_0)}z_i^{(\ell)}(t_0) + \sup_{\varsigma\in[t_0, t]}\big|\nu_{i,\ell}(\varsigma)\big|,
           \end{equation}
           holds for all $t\geq t_0$ and $i\in\mathcal{F}$, where $|\cdot|$ denotes the Euclidean norm of a vector and
           \begin{equation}\label{input_for_estimate}
           \nu_{i,\ell}(t) :=  \sum_{j=1}^{m} \frac{a_{ij}}{\kappa_i}z_j^{(\ell)}(t_{k_{ij}^{\mathrm{x}}(t)}) + \frac{1}{\kappa_i}\xi_{i,\ell}(t).
           \end{equation}
           Denote $|\bs{z}_\ell(s)|:= \big(|z_1^{(\ell)}(s)|, \ldots, |z_m^{(\ell)}(s)|\big)^\top$, $|\bs{\nu}_{\ell}(s)|:= \big(|\nu_{1,\ell}(s)|, \ldots, |\nu_{m,\ell}(s)|\big)^\top$, and let $\beta(|\bs{z}_{\ell}(t_0)|, t-t_0)$ be the vector stack of all the first exponentially decaying terms in the right-hand-side of \eqref{estimate_state} for $i\in\mathcal{F}$. Then, inequality \eqref{estimate_state} can be rewritten in the form
          \begin{eqnarray}
        |\bs{z}_\ell(t)| \leq \beta(|\bs{z}_{\ell}(t_0)|, t-t_0) + \Gamma_\ell^0 \cdot \sup_{\varsigma\in[t_0, t]}\big|\bs{\nu}_{\ell}(\varsigma)\big|,
        \end{eqnarray}
        for all $t\geq t_0$, where the supremum of a vector argument is understood in the element-wise sense. Consequently, one can conclude that the overall system that consists of all systems \eqref{dot_z}, for $i\in\mathcal{F}$, with the $m$ states/outputs $z_i^{(\ell)}$ and $m$ inputs $\nu_{i,\ell}$, $i\in\mathcal{F}$, is ISS/IOS and the IOS gain matrix $\Gamma^0_{\ell}$ is equal to $I_m$. Also, one can verify from \eqref{input_for_estimate} that
           \begin{equation}
            |\bs{\nu}_{\ell}(t)| \leq \mathcal{M}_\ell \cdot \sup_{\varsigma\in[t_{k_{ij}^{\mathrm{x}}(t)}, t]}\big|\bs{z}_\ell(\varsigma)\big| + \mathbf{D}^{-1}|\Xi_{\ell}(t)|,
           \end{equation}
        holds for all $t \geq 0$, where $|\Xi_{\ell}(t)|:= \big(|\xi_{1,\ell}(t)|, \ldots, |\xi_{m,\ell}(t)|\big)^\top$, $\mbf{D}:=\mbox{diag}\{d_{ij}\}$, with $d_{ii} = \kappa_i$, $i\in\mathcal{F}$, and the elements of the interconnection matrix $\mathcal{M}_\ell :=\{ \mu_{ij,\ell} \}\in\mathbb{R}^{m\times m}$ are obtained as $\mu_{ij,\ell} = \frac{a_{ij}}{\kappa_i}$, for $i, j\in\mathcal{F}$. It can be deduced from the assumptions of the proposition that $|\Xi_{\ell}(t)|$ is uniformly bounded and $|\Xi_{\ell}(t)|\to 0$. This, with the fact that $(t - t_{k_{ij}^{\mathrm{x}}(t)})$ is bounded by Assumption~\ref{AssumptionCommunicationBlackouts}, lead to the conclusion that $z_i^{(\ell)}$ and $\nu_{i,\ell}$, $i\in\mathcal{F}$, are uniformly bounded and converge asymptotically to zero under the condition that $\rho(\Gamma_\ell) < 1$, where $\rho(\Gamma_\ell)$ denotes the spectral radius of the closed loop gain matrix $\Gamma_\ell$ defined as $\Gamma_\ell:=\Gamma^0_\ell~ \mathcal{M}_\ell$ (see Theorem 2 in \cite{abdessameud2015synchronization}).

        In view of \eqref{laplacian}, matrix $\Gamma_\ell$ can be rewritten as
           $\Gamma_\ell = I_m - \mbf{D}^{-1} \mbf{L}_1$,
        where matrix $\mathbf{L}_1$ in \eqref{laplacian} is a nonsingular M-matrix by Lemma~\ref{lemma_Mei}, and so is $\mbf{D}^{-1} \mbf{L}_1 = I_m - \Gamma_{\ell}$. Consequently, one can conclude, from \cite[page 167]{qu2009cooperative}, that $\rho(\Gamma_\ell)<1$. As a result, $z_{i}^{(\ell)}$, $\nu_{i,\ell}$ are uniformly bounded and $z_{i}^{(\ell)}(t)\to 0$, $\nu_{i,\ell}(t)\to 0$ for $i\in\mathcal{F}$. The results of the proposition then follow using the definition of $z_{i}^{(\ell)}$ and  Assumption~\ref{assum:S}.

%
%

\end{document}